\documentclass[a4paper,11pt]{article}

\usepackage{amsfonts}
\usepackage{amsmath}
\usepackage{mathtools}
\usepackage{amsthm}
\usepackage{array}
\usepackage[nobreak]{cite}
\usepackage[usenames,dvipsnames]{color}
\usepackage{complexity}
\usepackage{environ}
\usepackage[margin=1in]{geometry}
\usepackage{graphicx}
\usepackage[utf8]{inputenc}
\usepackage{longtable}
\usepackage{subcaption}
\usepackage{tablefootnote}
\usepackage{url}
\usepackage{verbatim}
\usepackage[normalem]{ulem}
\usepackage{xspace}
\usepackage{listings}
\definecolor{gray}{rgb}{0.5,0.5,0.5}
\makeatletter
\def\lst@visiblespace{\hspace{0.70em}}
\makeatother

\usepackage{hyperref}

\let\oldenumerate\enumerate
\renewcommand{\enumerate}{
  \oldenumerate
  \setlength{\itemsep}{1pt}
  \setlength{\parskip}{0pt}
  \setlength{\parsep}{0pt}
}
\let\olditemize\itemize
\renewcommand{\itemize}{
  \olditemize
  \setlength{\itemsep}{1pt}
  \setlength{\parskip}{0pt}
  \setlength{\parsep}{0pt}
}

\newcommand{\find}[0]{\texttt{find}\xspace}
\newcommand{\mkdir}[0]{\texttt{mkdir}\xspace}

\newcommand{\calL}[0]{\mathcal{L}}
\newcommand{\code}[1]{\texttt{#1}}
\newcommand{\braces}[0]{\texttt{\symbol{123}\symbol{125}}\xspace}
\newcommand{\concat}{\,}
\newcommand{\bs}[0]{\symbol{92}}
\newcommand{\nul}[0]{{\normalfont\textsc{nul}}\xspace}
\newcommand{\inc}[0]{{\normalfont\textsc{inc}}\xspace}
\newcommand{\dec}[0]{{\normalfont\textsc{dec}}\xspace}
\newcommand{\jz}[0]{{\normalfont\textsc{jz}}\xspace}
\newcommand{\jump}[0]{{\normalfont\textsc{j}}\xspace}
\newcommand{\ispc}[0]{{\normalfont\textsc{ispc}}\xspace}
\newcommand{\sep}[0]{{\normalfont\textsc{sep}}\xspace}
\newcommand{\Sep}[0]{\overline{\sep}\xspace}
\newcommand{\shortgets}{\mathrel{\resizebox{0.75em}{!}{$\leftarrow$}}}
\newcommand{\pc}[0]{{\mathrm{pc}}\xspace}

\theoremstyle{definition}
\newtheorem{definition}{Definition}[section]

\newtheorem{theorem}[definition]{Theorem}

\title{Turing Completeness of GNU \texttt{find}: From \texttt{mkdir}-assisted Loops to Standalone Computation}
\author{Keigo Oka\thanks{Google. Work done in a personal capacity.}\\
\href{mailto:ogiekako@gmail.com}{ogiekako@gmail.com}}
\date{}

\begin{document}

\maketitle

\begin{abstract}
  The Unix command \texttt{find} is among the first commands taught to beginners, yet remains indispensable for experienced engineers.
  In this paper, we demonstrate that \texttt{find} possesses unexpected computational power, establishing three Turing completeness results using the GNU implementation (a standard in Linux distributions).
  (1) \texttt{find} + \texttt{mkdir} (a system that has only \texttt{find} and \texttt{mkdir}) is Turing complete: by encoding computational states as directory paths and using regex back-references to copy substrings, we simulate 2-tag systems.
  (2) GNU \texttt{find} 4.9.0+ alone is Turing complete: by reading and writing to files during traversal, we simulate a two-counter machine without \texttt{mkdir}.
  (3) \texttt{find} + \texttt{mkdir} without regex back-references is still Turing complete: by a trick of encoding regex patterns directly into directory names, we achieve the same power.

  These results place \texttt{find} among the ``surprisingly Turing-complete'' systems, highlighting the hidden complexity within seemingly simple standard utilities.
\end{abstract}

\section{Introduction}

The Unix command \find is among the first utilities introduced to students and system administrators.
While \find is designed to search for files in a directory hierarchy, its recursive traversal with the ability to execute commands opens the door to complex behaviors.
In this paper, we demonstrate that \find possesses unexpected computational power.\footnote{A preliminary version of this work appeared on the author's blog~\cite{ogiekako2024find} and attracted considerable interest~\cite{hackaday2024find}.}

Foundational work by Minsky~\cite{minsky1967computation} and Post~\cite{post1943formal} established that very simple systems, such as tag systems and counter machines, are capable of universal computation.
This insight explains why ``accidental'' Turing completeness is a recurring theme in computer science: systems designed for specific, limited purposes often turn out to have enough power to perform arbitrary computations.
Famous examples include the stream editor \code{sed}~\cite{sed-tc, blaess-turing-sed}, the x86 \code{mov} instruction~\cite{dolan2013mov}, and even the rules of \emph{Magic: The Gathering}~\cite{churchill2019magic}.
In security research, such emergent computational power has been termed ``weird machines''~\cite{bratus2013weird}.

\subsection{Contributions}

In this paper, we establish three Turing completeness results. Our proofs use the standard many-one reduction notion:
for each instance of a universal model,
we effectively construct an execution whose behavior simulates that instance.

\begin{enumerate}
  \item \textbf{GNU \find + \mkdir is Turing complete} (Section~\ref{sec:find_mkdir_tc}): By encoding computational states as directory paths and using regex back-references to copy substrings, we simulate 2-tag systems. We use options added in GNU \find 4.2.12, which was released in 2005.

  \item \textbf{GNU \find alone is Turing complete} (Section~\ref{sec:gnu_find_tc}): By reading and writing to files during traversal, we simulate a two-counter machine without \mkdir. This construction utilizes the \code{-files0-from} option added in GNU \find 4.9.0, which was released in 2022.

  \item \textbf{GNU \find + \mkdir without regex back-references is Turing complete} (Section~\ref{sec:no_backref_tc_sketch} (sketch), Appendix~\ref{sec:no_backref_tc}): By a technique of encoding regex patterns directly into directory names, we achieve the same power even without regex back-references.
\end{enumerate}

Our constructions use $O(1)$-length path components (hence do not assume an increased \code{NAME\_MAX}). In practice, resource limits such as storage capacity, maximum path length,
and finite inodes serve to cap how many simulation steps can be realized on a given machine.

These results place \find among the ``surprisingly Turing-complete'' systems, highlighting the hidden complexity within seemingly simple standard utilities.

We implement our constructions and confirm them on toy examples (Appendix~\ref{sec:toy-examples}).

\subsection{Applicability and Standards}

We restrict our attention to GNU \find, because (1) it is standard in Linux distributions, and (2) our constructions rely on options that are GNU extensions.
While the options we use for \find + \mkdir are also found in BSD \find (macOS/FreeBSD), subtle runtime differences prevent our constructions from working as-is, though they may work with some modifications.

Table~\ref{tab:implementations} summarizes the Turing completeness results established in this paper.

\begin{table}[h]
  \centering
  \small
  \begin{tabular}{|l|l|c|c|l|}
    \hline
    \textbf{\find Implementation} & \textbf{Key features} & \textbf{Model} & \textbf{Turing complete} & \textbf{Notes} \\ \hline
    GNU ($\ge 4.9.0$) & \code{-files0-from} & \find only & \textbf{Yes} & Sec.~\ref{sec:gnu_find_tc} \\ \hline
    GNU ($\ge 4.2.12$)  & \code{-regex} \code{-execdir} & \find + \mkdir & \textbf{Yes} & Sec.~\ref{sec:find_mkdir_tc} \\ \hline
    GNU ($\ge 4.2.12$) & \code{-regex} \code{-execdir} & \find + \mkdir & \textbf{Yes}\textsuperscript{\dag} & Sec.~\ref{sec:no_backref_tc_sketch} \\ \hline
    POSIX specification & - & \find + \mkdir & \textbf{Unknown} & No \code{-regex} \\ \hline
    \multicolumn{5}{l}{\footnotesize \textsuperscript{\dag} Proven even without using regex back-references.}
  \end{tabular}
  \caption{Turing completeness results for \find.}
  \label{tab:implementations}
\end{table}

The question of whether similar computational universality exists in other implementations, such as BSD, BusyBox, and uutils \find, remains an intriguing avenue for future research.

\section{Preliminaries}
\label{sec:preliminaries}

To establish that computation arises solely from \find and \mkdir, we define the following execution model.

\begin{definition}
  \label{def:execution-model}
  Let $\mathcal{B} = \{\find, \mkdir\}$ (for the \find + \mkdir system) or $\mathcal{B} = \{\find\}$ (for the \find system) be the set of permissible binaries.
  Fix a fresh, writable working directory ``.'' that is initially empty.

  An \emph{execution} is a finite sequence of command invocations $C_0, \cdots, C_{m-1}$,
  where each $C_i$ is an argv vector specified in advance.
  (The sequence is non-adaptive in the sense that $C_0, \cdots, C_{m-1}$ are fixed before the execution starts.)
  The execution is valid if and only if every spawned process has its executable image in $\mathcal{B}$.
  The \emph{result} of the execution is the concatenation of the standard output of the commands in execution order.
  We say the execution \emph{halts} if all commands in the sequence terminate.
\end{definition}

This condition can be intuitively understood through a security lens: consider a restricted, shell-free system (e.g., a minimal container) that allows an adversary to execute commands starting from $\mathcal{B}$ a constant number of times. What we prove is that even if $\mathcal{B}$ are the only executables in the system, the adversary can still force the system to perform arbitrary computations.

We clarify the resource model used throughout this paper.

\begin{definition}[Resource model]\label{def:resource-model}
  We work in an idealized file-system model with unbounded capacity: there is no fixed upper bound on the number of directory entries that can be created, no bound on the total path length, and no bound on the amount of data stored in file contents.
  We do not, however, assume unbounded length of a single path component (i.e., a file basename): our constructions use only $O(1)$-length path components and therefore do not rely on increasing \code{NAME\_MAX}, which we treat as a fixed finite constant (e.g., 255 on many systems).
\end{definition}

Before presenting our constructions, we review the key concepts of GNU \find~\cite{gnu-findutils} and \mkdir.

The \find utility recursively traverses a directory tree, evaluating an \emph{expression} for each file visited. This expression functions as a domain-specific language for file filtering and manipulation.
The expression is composed of:

\begin{itemize}
  \item \textbf{Tests}: Predicates that return true or false based on file attributes (e.g., \code{-name}, \code{-empty}, \code{-regex}).
  \item \textbf{Actions}: Operations that perform side effects (e.g., \code{-exec}, \code{-delete}) or control traversal (e.g., \code{-prune}, \code{-quit}).
  \item \textbf{Operators}: Logical connectives \code{-and} (implicit), \code{-or} (\code{-o}), and \code{-not} (\code{!}), where and/or evaluations are short-circuiting (right-hand side is evaluated only if needed).
\end{itemize}

Of particular importance are the actions \code{-exec} and \code{-execdir}.
They accept a command template terminated by \code{;}, where the string \braces is replaced by the current file path. The specified (either \find or \mkdir) command is executed directly (via fork and exec) without invoking a shell, adhering to our execution model. The action is evaluated to true only if the exit status of the command is 0. \code{-exec} runs the command in the same directory as the outer \find, and \code{-execdir} runs the command in the directory holding the current file, with \braces replaced with the filepath relative from the directory. This locality of \code{-execdir} is vital for our scale-independent loop construction (Section~\ref{sec:find_mkdir_impl}).

We note that nesting of \code{-exec} (e.g., \code{-exec find ... -exec ... ; ;}) is syntactically invalid, because the outer \find consumes the first \code{;} as its own terminator, leaving the inner \code{-exec} without a terminator. This restriction significantly constrains our design.

The directory traversal order of \find is pre-order by default, meaning a directory is visited before its children. This can be changed with a global option \code{-depth}, which changes the traversal order to post-order.
Table~\ref{tab:find_options} summarizes the \find options used in this paper.

\code{mkdir \textit{path}} creates a directory at the path. If its parent directory does not exist, creation fails and the command exits with a non-zero status. \code{mkdir -p \textit{path}} creates the directory and its parent directories if they do not exist, and exits with a zero status. These functionalities are specified in POSIX and consistent across platforms.
Our construction never relies on creation failure due to permission issues.

\begin{table}[ht]
  \centering
  \begin{tabular}{lp{0.58\textwidth}l}
    \hline
    \textbf{Option/Operator} & \textbf{Description} & \textbf{Added in}\\
    \hline
    \code{-true} & True. & - \\
    \code{-empty} & True if the file/directory is empty. & - \\
    \code{-name \textit{pattern}} & True if the file's basename matches the pattern. & - \\
    \code{-size \textit{n}c} & True if the file size is \textit{n} bytes. & - \\
    \code{-regex \textit{pattern}} & True if the file path matches the pattern.  & - \\
    \code{-exec \textit{cmd} ;} & Executes \textit{cmd} with \braces replaced by the path. True if the exit status is 0. & - \\
    \code{-execdir \textit{cmd} ;} & Executes \textit{cmd} in the file's directory with \braces replaced by the relative path. True if the exit status is 0. & 4.2.12 \\
    \code{-prune} & True. Do not descend into the matched directory. & - \\
    \code{-quit} & Terminates the entire \find execution. & - \\
    \code{-delete} & Removes the matched file or (if empty) directory. Implies \code{-depth}. True if removal succeeds. & 4.2.3 \\
    \code{-printf \textit{fmt}} & True. Print formatted output to standard output. & - \\
    \code{-fprintf \textit{file} \textit{fmt}} & True. Write formatted output to \textit{file}. & - \\
    \code{-depth} & Global option to do post-order traversal. & - \\
    \code{-files0-from \textit{file}} & Global option to read starting points from \textit{file} (NUL-separated) instead of the command line. & 4.9.0 \\
    \code{( \textit{expr} )} & True if \textit{expr} is true. Force precedence. & - \\
    \code{! \textit{expr}} & Not (true if \textit{expr} is false). & - \\
    \code{\textit{expr1} \textit{expr2}} & And. \textit{expr2} is not evaluated if \textit{expr1} is false. & - \\
    \code{\textit{expr1} -o \textit{expr2}} & Or. \textit{expr2} is not evaluated if \textit{expr1} is true. & - \\
    \hline
  \end{tabular}
  \caption{Summary of GNU \find options used in this paper. Operators below \code{( \textit{expr} )} are listed in order of decreasing precedence.}
  \label{tab:find_options}
\end{table}

\section{Turing Completeness of \find + \mkdir}
\label{sec:find_mkdir_tc}

In this section, we show the following theorem.

\begin{theorem}
  \label{thm:find_mkdir_tc}
  GNU \find + \mkdir is Turing complete for \find version $\ge$ 4.2.12 (as of this writing when 4.10.0 is the latest), assuming the resource model in Definition~\ref{def:resource-model}.
\end{theorem}

\subsection{Loop Construction}
\label{sec:find_mkdir_impl}

The foundation of our constructions is the ability to create an infinite loop using only \find and \mkdir. For this, a property of \find we rely on is that child directories created during traversal of a directory become visible and are visited in the same execution if the traversal is pre-order (default). When \find visits a directory $D$ and an action creates a new subdirectory $C$ inside $D$, $C$ is visible to the subsequent directory enumeration and will be visited in the same \find execution.
While POSIX leaves this behavior unspecified\footnote{``If a file is removed from or added to the directory hierarchy being searched it is unspecified whether or not find includes that file in its search.''~\cite{posix-find}}, this behavior is consistent across GNU \find versions\footnote{Verified on GNU find 4.10.0 (Debian) and 4.5.11 (CentOS 7).}.

The following code implements an infinite loop. In our code snippets, we do not escape characters like \code{;} (e.g., as \verb|\;|), which is typically required by a shell. This is to signify shell independence. Note that proper escaping is necessary to execute these snippets in a shell.

\begin{verbatim}
mkdir x
find x -execdir mkdir {}/x ;
\end{verbatim}

This code creates \code{x}, then \find visits \code{x} and executes \code{mkdir} command in the directory \code{.} (the parent of \code{x}) with the argument \code{./x/x}, because \braces is substituted with \code{./x}, that is the relative path of \code{x} from \code{.} . This creates \code{x/x}. Due to the aforementioned property, \find then visits \code{x/x}, and it then creates \code{x/x/x}, and so on indefinitely.

We use \code{-execdir}, which was added in GNU \find 4.2.12.
While strictly speaking our resource model (Definition~\ref{def:resource-model}) does not forbid arbitrary long arguments, using \code{-execdir} is practically crucial.
If we use \code{-exec}, \find passes the path of the current file to the command.
As the directory structure grows deeper, this path eventually exceeds \code{PATH\_MAX} (typically 4096 bytes), causing \code{mkdir} to fail and \find to halt. \code{-execdir} avoids this by executing the command in the subdirectory containing the current file, allowing us to pass an argument of constant length.

\subsection{Tag System Implementation}

We prove that \find + \mkdir is Turing complete by simulating a 2-tag system \cite{post1943formal}. The following is a definition of a 2-tag system equivalent to the one given by Rogozhin~\cite{rogozhin1996small}.

\begin{definition}
  \label{def:2-tag-system}
  A \emph{2-tag system} is defined by a pair $(\Sigma,P)$, where $\Sigma$ is a finite alphabet containing a \emph{halting symbol} $H\in\Sigma$, and $P$ is a (partial) map with domain $\Sigma\setminus\{H\}$ and codomain $\Sigma^*$.
  A \emph{word} is a finite string over $\Sigma$. The computation starts with an \emph{initial word} $w_1\in\Sigma^*$, and proceeds iteratively. In the $i$-th iteration, with $w_i = x_{i,1}\cdots x_{i,k}$, the computation halts if $k<2$ or $x_{i,1}=H$. Otherwise, the next word is $w_{i+1} = x_{i,3}\cdots x_{i,k}\,P(x_{i,1})$.
  If the computation halts at $w_t$, $w_t$ is called the \emph{halting word}, and is the \emph{output} of the computation.
\end{definition}

It is known that 2-tag systems are computationally universal.
It is known that a universal 2-tag system exists~\cite{minsky1961tag}, and moreover, combining a reduction from 2-symbol $m$-state Turing machines
to 2-tag systems with a universal machine in $\mathrm{TM}(18,2)$~\cite{neary2009four} yields a universal tag system in $\mathrm{TS}(576,2)$ (i.e., a universal 2-tag system with $576$ symbols)~\cite{demol2008turing}.

We prove Turing completeness of \find + \mkdir by simulating an arbitrary 2-tag system in $\mathrm{TS}(576,2)$. We will encode the input word to an ASCII string and decode the halting word from an ASCII string, but since this mapping is straightforward (clearly computable by a deterministic finite state transducer), we omit a formal proof for the validity of the encoding.
We use \code{typewriter} font to denote string literals, simple juxtaposition for string concatenation, and $\varepsilon$ for the empty string.

Let $\mu = 576$, and $\Sigma=\{s_1,\dots, s_\mu, s_{\mu+1} = H\}$. Let $A$ be a set of lowercase ASCII letters, and $\{s_1', \dots,s_\mu', s_{\mu+1}' \}\subset A^2$ be distinct strings with length 2 (possible because $577 < 26^2$).
We define $\sigma_k = \code{/}\concat s_k'$, where \code{/} is the ASCII letter for directory separator. We let $\Sigma'=\{\sigma_1,\dots,\sigma_\mu, \sigma_{\mu+1} = \eta\}$ and $\phi: \Sigma \to \Sigma'$ be a function that maps $s_i$ to $\sigma_i$. We also define $\Phi: \Sigma^* \to \Sigma'^*$ so that $\Phi(\varepsilon) = \varepsilon$ and $\Phi(s_1\cdots s_k) = \Phi(s_1\cdots s_{k-1}) \concat \phi(s_k)$ for $k \ge 1$. We define $\pi(s_i) = \Phi(P(s_i))$ for $i \le \mu$.

We first run the following command to embed the initial word. We use spaces to separate arguments, distinguishing them from juxtaposition (concatenation).
\begin{lstlisting}
mkdir -p _$\concat\Phi(w_1)\concat$/_
\end{lstlisting}
The idea for the following computation is to iteratively append the file path representing the next word to it using \code{/\_} as a separator. During the computation, the working directory should contain exactly one directory \code{\_}, each of whose descendants is a directory containing at most one directory. We call the path to the empty directory at a point of the computation the state of the system at that point.
After the computation halts, the state of the system should be $\code{\_}\concat\Phi(w_1)\concat\code{/\_}\concat\Phi(w_2)\concat\code{/\_}\cdots\code{/\_}\Phi(w_t)\concat\code{/\_}$, and during the computation the state should be a prefix of it. Let $\Theta_{i,j}$ be the following state. Our initial state is $\Theta_{1,0}$ and the last state will be $\Theta_{t,0}$.
$$\Theta_{i,j} = \code{\_} \cdots \code{\_}\concat\Phi(w_i)\concat\code{/\_}\concat\phi(x_{i+1,1})\cdots\phi(x_{i+1,j})$$
At $\Theta_{i, j}$, if $j=0$ and $w_i$ is a halting word, we stop the computation. Otherwise, (1) if $\phi(x_{i+1,1})\cdots\phi(x_{i+1,j})$ equals $\phi(x_{i,3})\cdots\phi(x_{i,j+2})$ and $x_{i,j+3}$ exists, we append $\phi(x_{i,j+3})$ to the state, and (2) otherwise (if we have copied the last alphabet of the previous word), we append $\pi(x_{i,1})\code{/\_}$ to the state.

Let $\lambda = \code{/[a-z][a-z]}$ be a regex that matches any string in $\Sigma'$. \code{[a-z]} matches any lowercase ASCII letter, but not \code{/} and \code{\_}. Let $\Lambda=\code{\bs(}\concat \lambda \concat \code{\bs)*}$ be a regex that matches any repetition of strings in $\Sigma'$. \code{\bs(} and \code{\bs)} are verbatim because a backslash before parentheses is required by the \find's default regex type (emacs). Same for \code{\bs|} below. Define the following strings.
\begin{align*}
  \alpha_k &= \code{.*\_} \concat \lambda \concat \lambda \concat \code{\bs(} \concat \Lambda \concat \code{\bs)} \concat \sigma_k \concat \Lambda \concat \code{/\_\bs1} \\
  \beta_k  &= \code{.*\_} \concat \sigma_k \concat \Lambda \concat \code{/\_} \concat \Lambda \\
  \gamma &= \code{.*\_\bs(\bs|} \concat \lambda \concat \code{\bs|} \concat \eta \concat \Lambda \concat \code{\bs)/\_} \\
\end{align*}
We run the following as the second command. Newlines are for readability and have the same meaning as spaces.
\begin{lstlisting}
find _ -empty (
  -regex $\gamma$ -quit -o
  -regex $\alpha_1$ -execdir mkdir {}$\concat \sigma_1$ ; -o
  $\cdots$
  -regex $\alpha_{\mu+1}$ -execdir mkdir {}$\concat \sigma_{\mu+1}$ ; -o
  -regex $\beta_1$ -execdir mkdir -p {}$\concat\pi(s_1)\concat$/_ ; -o
  $\cdots$
  -regex $\beta_{\mu}$ -execdir mkdir -p {}$\concat\pi(s_{\mu})\concat$/_ ; -o
  -printf unreachable
)
\end{lstlisting}
Finally we run the following command to output the encoded result.
\begin{lstlisting}
find _ -depth ! -empty -name _ -execdir find _ ! -name _ -printf /%f ; -quit
\end{lstlisting}
We show that the output of the last command is $\Phi(w_t)$ if the original 2-tag system halts with $w_t$, and the second command does not halt if the original does not halt, proving Theorem~\ref{thm:find_mkdir_tc}.

\begin{proof}
  Let $\calL(r)$ be the language denoted by the regex $r$. Remember that $\Sigma' \subset \calL(\lambda), \Sigma'^* \subset \calL(\Lambda)$ and any string including \code{\_} or ending with \code{/} is not in $\calL(\lambda)$ nor $\calL(\Lambda)$. $\{\sigma_k\} = \calL(\sigma_k)$ and $\{\eta\} = \calL(\eta)$. We have $\sigma_k=\phi(s_k)$ and $\Theta_{i,j} = \cdots\code{\_}\concat\phi(x_{i,1})\cdots\phi(x_{i,|w_i|})\code{/\_}\concat\phi(x_{i+1,1})\cdots\phi(x_{i+1,j})$, where $x_{i+1,j+2} = x_{i,j}$ for $1 \le j\le |w_i|-2$.

  We prove that the state is $\Theta_{t,0}$ if the second command halts by induction. Let us call the expression in the parentheses the \emph{main expression}.
  As the base case, \find visits $\Theta_{1,0}$ and evaluates the main expression for the first time because of the \code{-empty} predicate. Assume the current state is $\Theta_{i,j}$ and \find visits it, evaluating the main expression. By construction, $\gamma$ matches if and only if $j=0$ and $|w_i| < 2$ or $x_{i,1} = H$, meaning $w_i$ is the halting word. If $j < |w_i| - 2$ and $x_{i,j+3} = s_k$, $\Theta_{i,j}\in \calL(\alpha_k)$ because $\sigma_k$ matches with $\phi(s_k)$, and $\code{\bs(}\concat\Lambda\concat\code{\bs)}$ and $\code{\bs 1}$ match with $\phi(x_{i,3})\cdots\phi(x_{i,j+2})$ and $\phi(x_{i+1,1})\cdots\phi(x_{i+1,j})$ which are the same string. The converse can be said because the $\code{\_}$ before $\code{\bs 1}$ combined with the fact the same string should match $\Lambda$ constrains it to match $\phi(x_{i+1,1})\cdots\phi(x_{i+1,j})$ and $\sigma_k$ only matches with $\phi(s_k)$. Provided $j = |w_i| - 2$, $\Theta_{i,j}\in \calL(\beta_k)$ if and only if $\phi(x_{i,1}) = \sigma_k$ by construction. These enforce the main expression to work as follows: (1) If $(i,j) = (t,0)$, it halts the command. (2) Otherwise if $j < |w_i| - 2$, it executes $\code{mkdir \braces}\concat\phi(x_{i,j+3})$ turning the state into $\Phi_{i,j+1}$. (3) Otherwise at $j = |w_i|-2$, it executes $\code{mkdir -p \braces}\concat\pi(x_{i,1})\code{/\_}$ turning the state into $\Phi_{i+1,0}$. The newly created state is subsequently visited due to the said GNU \find property, so the induction completes.

  The third command visits the state in post-order (\code{-depth}), matches the first non-empty \code{\_}, which is the second-to-last \code{\_} directory (say \emph{base}), executes the inner \find for the base directory, which visits its descendants in pre-order, printing the name of each directory prefixing \code{/} (due to \code{-printf /\%f}) unless the name is \code{\_} (due to \code{! -name \_}), effectively outputting $\Phi(w_t)$, then quit.
\end{proof}

\section{Turing Completeness of GNU \find (Standalone)}
\label{sec:gnu_find_tc}

We prove that GNU \find version 4.9.0 or later is Turing complete \emph{without} \mkdir.
We prove this by simulating Minsky's program machine \cite{minsky1967computation}.

\begin{theorem}
  \label{thm:gnu_find_tc}
  GNU \find is Turing complete for \find version $\ge$ 4.9.0 (as of this writing when 4.10.0 is the latest), assuming the resource model in Definition~\ref{def:resource-model}.
\end{theorem}

\subsection{Loop Construction}

The foundation of our constructions is the ability to create an infinite loop using only \find. We utilize \code{-files0-from \textit{file}}, which lets \find read the starting points from a file separated by \nul instead of the command line, where \nul is the ASCII character represented by a byte with value 0. \find's \code{-fprintf} can write \nul to a file provided \code{\bs 0} verbatim.
We update the file while the \find command is running letting infinite starting points be supplied to the command, causing a loop.

We first introduce a way of encoding a non-negative integer as a file. Let $\iota = \code{.}\concat\nul$ (string with length 2).
If the content of a file $f$ is the string $\iota$ repeated $n$ times, we consider the value of $f$ to be $n$. If $f$ does not exist, we consider the value to be 0. In other cases we consider the value of the file undefined. We will use files \code{a}, \code{b}, \code{s}, \code{t} and represent their values with $a, b, s, t$ respectively. We use $t$ for temporary storage. We construct our program so that throughout the computation $a,b,s$ are always defined.

Notice that the code like \code{find -files0-from s -prune \textit{expr}} evaluates \textit{expr} $s$ times, and it exits with non-zero exit code if $s=0$ (\code{s} does not exist).

For loop construction, we first initialize $s$ (for stream) with 1.
\begin{lstlisting}
find -fprintf s .\0 -quit
\end{lstlisting}

We define an expression which increments $s$ when evaluated as follows, and call this $\inc(\code{s})$. We similarly define $\inc(\code{a})$ and $\inc(\code{b})$ for later use, in which \code{s} is replaced with \code{a} or \code{b}.
\begin{lstlisting}
( ( -exec find s -quit ;
-exec find -quit -fprintf first . ;
-exec find -files0-from s ( -name . -o -name first -delete ) -fprintf t .\0 ;
-exec find -files0-from t -prune -fprintf s .\0 ;
-exec find t -delete ; ) -o
-exec find -fprintf s .\0 -quit ; )
\end{lstlisting}

To understand this, remember that the file given to \code{-fprintf} is truncated or created to 0 bytes as soon as the \find command starts, regardless  of whether its predicate ever matches (as documented). Line 1 is evaluated to true if and only if $s \neq 0$. Line 2 (where $s>0$) creates an empty file \code{first} (because of \code{-quit}, \code{-fprintf} is not evaluated). Line 3 updates \code{t} with \code{.}\concat\nul repeated $s+1$ times and removes \code{first} (first evaluation matches with both \code{.} and \code{./first}, removing \code{first} with \code{-delete}), effectively meaning $t \shortgets s+1$. Line 4 updates \code{s} with \code{.}\concat\nul repeated $t$ times ($s \shortgets t$), and line 5 deletes \code{t} ($t \shortgets 0$). If $s=0$, line 6 updates \code{s} with \code{.}\concat\nul ($s \shortgets 1$).

The loop is constructed as follows. Remember we start with $s=1$.
\begin{lstlisting}
find -files0-from s -prune $\inc(\code{s})$;
\end{lstlisting}

This works because GNU \find reads \code{s} as a stream. When the first \code{.}\concat\nul is read, \find still has not read the EOF (end of file). It then executes $\inc(\code{s})$ to extend the file so that it still has the next \code{.}\concat\nul, and this process continues indefinitely. The fact that this works depends entirely on the implementation of GNU \find, but this has been empirically confirmed on GNU \find 4.10.0 and rigorously confirmed to work by reading the source code of GNU \find 4.9.0. 

\subsection{Program Machines}

The following definition for 2-counter program machine is equivalent to the one given by Minsky~\cite{minsky1967computation}. 2-counter program machines are known to be Turing complete.

\begin{definition}
  \label{def:2-counter}
  A 2-counter program machine is a finite sequence of instructions \\
  $P=(I_0,\dots,I_{m-1})$ acting on counters $c_0,c_1\in\mathbb{N}$.
  A configuration is $(\pc,c_0,c_1)\in\{0,\dots,m\}\times\mathbb{N}^2$,
  where $\pc=m$ denotes the \emph{halted} configuration. The one-step transition relation $\to_P$ is defined for $\pc<m$ as follows
  (with $r\in\{0,1\}$ and $q\in\{0,\dots,m\}$):
  \begin{itemize}
    \item If $I_\pc=\mathrm{INC}(r)$ then $c_r\gets c_r+1$ and $\pc\gets \pc+1$.
    \item If $I_\pc=\mathrm{DEC}(r)$ then
      $c_r\gets \max\{c_r-1,0\}$ and $\pc\gets \pc+1$.
    \item If $I_\pc=\mathrm{JZ}(r,q)$ then
      $\pc\gets q$ if $c_r=0$, and $\pc\gets \pc+1$ otherwise.
    \item If $I_\pc=\mathrm{J}(q)$ then $\pc\gets q$.
  \end{itemize}

  An execution of $P$ from an initial configuration $(0,c_0,c_1)$ is the maximal $\to_P$ chain; it \emph{halts} if and only if it reaches $\pc=m$, and outputs $c_0$.
\end{definition}

To simulate 2-counter program with \find, we create expressions to simulate the instructions. We have already seen $\inc$ expressions. We define the following expression as $\dec(\code{a})$, and $\dec(\code{b})$ analogously.
\begin{lstlisting}
( ( -exec find a -size 2c -delete ;
-exec find a -quit ;
-exec find -quit -fprint first ;
-exec find -files0-from a -name first -delete -fprintf t skip -o -name t -fprintf t . -o -name . -fprintf t \0 ;
-exec find -files0-from t -name . -fprintf a .\0 ;
-exec find t -delete ;
) -o -true )
\end{lstlisting}

When the expression is evaluated, it sets $a$ to $\max(a-1,0)$ as follows: Line 1 is evaluated to true and deletes \code{a} ($a\shortgets 0$) if and only if $a=1$ (due to \code{-size 2c}). Line 2 evaluates to false if $a=0$. Line 3 (with $a\ge 2$) creates \code{first}. Line 4 sets $t$ to $a-1$, which happens as follows: (1) Due to \code{-fprintf t}, \code{t} exists as an empty file as soon as the \find starts. (2) In the first iteration of \code{.} from \code{a}, because of \code{-delete} implying \code{-depth}, \code{./first} (being deleted) and \code{./t} are iterated before \code{.} and it causes either \code{.skip\concat\nul} or \code{skip.\concat\nul} to be written to \code{t}. (3) In the following $a-1$ iterations, since \code{first} no longer exists, $a-1$ repeat of \code{.}\concat\nul is appended to \code{t}. In line 5, since neither file \code{.skip} nor \code{skip.} exists, it is skipped with an error message written to standard error (which we do not use), causing \code{-fprintf a .\bs 0} to be evaluated $a-1$ times ($a\shortgets a-1$). Line 6 deletes \code{t} ($t\shortgets 0$). Line 7 exists to make $\dec(\code{a})$ always true.

The program counter is represented by the file \code{pc}, but unlike $a$ and $b$, we directly encode the value $n$ as a string of $n$ ones (ASCII character \code{1}). For example, the program counter being 3 means \code{pc}'s content is \code{111}. \code{pc} will be initialized to an empty file with \code{find -quit -fprintf pc not-written}, and the program counter will always be defined hereafter. Let $\code{1}^n$ denote a string of $n$ ones.

We define the expression $\jump(q)$ for $q\in\{0,\dots,m\}$ as follows. If an empty string is a concern, we can make it \code{-exec find -quit -fprintf pc x ;} for $q=0$. It is clear that it is evaluated to true and sets $\pc$ to $q$.
\begin{lstlisting}
-exec find -fprintf pc $\code{1}^q$ -quit ;
\end{lstlisting}
We define the expression $\jz(\code{a}, q)$ for $q\in\{0,\dots,m\}$ as follows, and $\jz(\code{b}, q)$ analogously. It is clear that $\jz(\code{a}, q)$ is evaluated to true and sets $\pc$ to $q$ if and only if $a=0$.
\begin{lstlisting}
( ( ! -exec find a -quit ; $\jump(q)$ ) -o -true )
\end{lstlisting}

Finally, we define the expression $\ispc(q)$ for $q\in\{0,\dots,m-1\}$ as follows.
Unlike other expressions, it assumes the current file is \code{pc}, and is evaluated to true if and only if $\pc = q$. $q$ in the code is substituted to the corresponding decimal number represented in ASCII.
\begin{lstlisting}
-size $q$c
\end{lstlisting}

The commands to simulate a given 2-counter program $P=(I_0,\dots,I_{m-1})$ and initial configuration $(0, c_0, c_1)$ is as follows. We let $\gamma_0 = \code{a}$ and $\gamma_1 = \code{b}$.

Our first command to run is \code{find -quit -fprintf pc x}, which initializes $\pc$ to 0. The next command to run is the following, where $\inc(\code{a})$ is repeated $c_0$ times and $\inc(\code{b})$ is repeated $c_1$ times. This initializes $s$, $a$, and $b$ with 1, $c_0$, and $c_1$ respectively.

\begin{lstlisting}
find $\inc(\code{s})$ $\inc(\code{a})\code{ }\cdots\code{ }\inc(\code{a})$ $\inc(\code{b})\code{ }\cdots\code{ }\inc(\code{b})$ -quit
\end{lstlisting}

The next command is the main loop. Let $\phi$ be a mapping from $\{I_0,\dots,I_{m-1}\}$ to \find expressions defined by the following.
\begin{align*}
  \phi(\mathrm{INC}(r)) &= \inc(\gamma_r)\\
  \phi(\mathrm{DEC}(r)) &= \dec(\gamma_r)\\
  \phi(\mathrm{JZ}(r, q)) &= \jz(\gamma_r, q)\\
  \phi(\mathrm{J}(q)) &= \jump(q)
\end{align*}
The main loop is as follows.
\begin{lstlisting}
find -files0-from s -name pc $\inc(\code{s})$ (
  $\ispc(0)$ $\jump(1)$ $\phi(I_0)$ -o
  $\cdots$
  $\ispc(m-1)$ $\jump(m)$ $\phi(I_{m-1})$ -o
  -quit
)
\end{lstlisting}

Finally, the following two commands allow us to directly output the value of $a$. The first command writes $a$ \code{1}s to \code{count}. The second command outputs the byte count of \code{count}, which is $a$.

\begin{verbatim}
find -files0-from a -fprintf count 1 -prune
find count -printf %s
\end{verbatim}

We show that by running the above commands from the first to the last, we get the same output as the original 2-counter machine if the machine halts, and it does not halt if the machine does not halt, completing the proof of Theorem \ref{thm:gnu_find_tc}. It is enough to show that the main loop works as expected.

\begin{proof}[Proof for Theorem \ref{thm:gnu_find_tc}]
  We call the expression in the parenthesis the \emph{main expression}. Until \code{-quit} in the main loop is evaluated, the main loop runs indefinitely: because no expression deletes \code{pc}, $\code{-name pc }\inc(\code{s})$ evaluates to true once for each starting point from \code{s}, appending \code{.\nul} to \code{s}. The main expression follows and is evaluated with the current file \code{pc}.

  We prove the claim by induction on the number of steps $t$ of the program machine.
  Let $(\pc^{(t)}, c_0^{(t)}, c_1^{(t)})$ be the configuration of the program machine at step $t$. Assume that the values encoded by \code{pc}, \code{a}, \code{b} ($\pc, a, b$) are equal to $\pc^{(t)}, c_0^{(t)}, c_1^{(t)}$ just before the $(t+1)$-th evaluation of the main expression. The base case $t=0$ holds as previously shown.

  At the $(t+1)$-th evaluation, $\ispc(i)$ is true if and only if $i = \pc = \pc^{(t)}$. Therefore, \code{-quit} is evaluated if and only if $\pc = m$, in which case both the main loop and the program machine halt, leaving $a=c_0^{(t)}$. Otherwise, the program machine executes $I_{\pc}$, and the \find command evaluates $\jump(\pc+1)\code{ }\phi(I_{\pc})$. This sequence updates $(\pc, a, b)$ exactly as $I_{\pc}$ updates $(\pc^{(t)}, c_0^{(t)}, c_1^{(t)})$, so the new values equal $(\pc^{(t+1)}, c_0^{(t+1)}, c_1^{(t+1)})$. Since $\jump(\pc+1)\code{ }\phi(I_{\pc})$ always evaluates to true, the main expression completes (due to the subsequent \code{-o}).
  This completes the proof.
\end{proof}

\section{Turing Completeness of \find + \mkdir without Back-references (Sketch)}
\label{sec:no_backref_tc_sketch}

While Section~\ref{sec:find_mkdir_tc} relies on regex back-references for the ``copy'' operation, we can achieve Turing completeness \emph{without} back-references.
We present the core ideas for our proof and provide a proof sketch. The full proof is in Appendix~\ref{sec:no_backref_tc}.

To signify non-use of back-references, we will always use \code{-regextype awk} global option for any \find that uses \code{-regex \textit{pattern}}. The awk regex engine does not support back-references. (It also uses unescaped parentheses \code{(} and \code{)} for grouping, unlike the default emacs regex engine which requires \code{\bs(} and \code{\bs)}, though this is not an important advantage for us, and the similar construction should work for any other regex types.)

We build on the idea in Section~\ref{sec:find_mkdir_tc} to simulate 2-tag systems. That is, during the computation, there is only one empty directory under the current working directory, and the path to it represents the state of the computation. The encoded symbols are strings $\{\sigma_i\}$ and when the computation halts, the last state encodes the words $w_1,\dots, w_t$ from left to right separated by separator strings. During the computation, the state has the form of $\Theta_{i,j}$, which is a prefix of the last state and encodes the status where the word $w_i$ is given and the first $j$ letters of the word $w_{i+1}$ have been computed. $x_{i,j}$ denotes the $j$-th symbol of $w_i$, and we let $\phi$ be a function that maps $s_i$ to $\sigma_i$, where $\{s_i\}$ are the symbols of the original tag-system.

We define $\Sigma'=\{\sigma_1, \dots,\sigma_{\mu+1}\}$ as before. That is, $\Sigma' \subset \calL(\code{/[a-z][a-z]})$ and $\sigma_i \neq \sigma_j$ for $i \neq j$.
The main trick is that we use $\sep=\code{\$\_)?(}$ as the separator, instead of \code{\_}. Those are allowed characters for filenames. Note that the letter $\code{\$}$ in regex is special in that it only matches the end of the string,  and $\code{?}$ indicates one or zero occurrence of the previous atom. This means $\code{\$\_}$ is a pattern that never matches any string.
This choice makes the state $\Theta_{i,j}$ during the computation look like the following.
$$
\Theta_{i,j} = \code{\$\_)?(}\cdots \concat \code{/\$\_)?(} \concat \phi(x_{i,1})\concat  \cdots\concat  \phi(x_{i,|w_i|})\concat  \code{/\$\_)?(}\concat  \phi(x_{i+1,1}) \concat\cdots\concat \phi(x_{i+1,j})
$$
Let us consider the regex \code{(}\concat$\Theta_{i,j}$\concat\code{)}, which should look like this:
$$
\code{(\$\_)?(}\concat\cdots\concat\code{\$\_)?}\concat\cdots\concat\code{(}\concat\cdots\concat\code{\$\_)?(}\concat  \phi(x_{i+1,1}) \concat\cdots\concat \phi(x_{i+1,j})\concat\code{)}
$$
This (valid) regex matches only $\phi(x_{i+1,1}) \concat\cdots\concat \phi(x_{i+1,j})$. The regex obviously matches the string, and conversely, if it matches (because $\calL(\cdots\concat\code{\$\_}) = \emptyset$), every $\code{(}\cdots\code{)?}$ must be matching with an empty string.

This observation allows us to update the part where we used back-references in the previous proof. That is, at state $\Theta_{i,j}$ where $j < |w_i|-2$, we append $\phi(x_{i,j+3})$ to the state without back-references. The idea is (1) we create $\mu+1$ marker files $t_1,\dots, t_{\mu+1}$ under the directory that represents the current state, (2) we invoke $\mu+1$ inner \find's so that $k$-th one's regex matches with $t_k$'s path only when $\phi(x_{i,j+3})=\sigma_k$, deleting the marker file. To construct such a regex, we embed the state into the regex with \braces (detailed below). (3) we check marker files and create a corresponding directory if a file is gone, advancing the state.

The regex constructed in the step (2) for $k$ is the following, where $\lambda=\code{/[a-z][a-z]}$ and $\Lambda=\code{(}\concat\lambda\concat\code{)*}$.
$$
\code{.*}\concat\sep\concat\lambda\concat\lambda\concat\code{(\braces)}\concat\sigma_k\concat\Lambda\concat\sep\Lambda\concat \code{/}\concat t_k
$$

The state $\Theta_{i,j}$ has $\phi(x_{i+1,1}) \concat\cdots\concat \phi(x_{i+1,j})$ after the separator. The \code{(\braces)} is expanded to $\code{(}\concat\Theta_{i,j}\concat\code{)}$, and as we have seen it matches with and only with $\phi(x_{i+1,1}) \concat\cdots\concat \phi(x_{i+1,j})$. We know it is equal to $\phi(x_{i,3})\concat\cdots\concat\phi(x_{i,j+2})$. This ensures that the regex matches if and only if $\phi(x_{i,j+3}) = \sigma_k$, and the current file is the marker file $t_k$.

We can use the same construction as the previous proof on other parts of the commands, and this completes the proof sketch of the following theorem.

\begin{theorem}
  \label{thm:no_backref_tc}
  GNU \find + \mkdir is Turing complete without using regex back-references for \find version $\ge$ 4.2.12 (as of this writing when 4.10.0 is the latest), assuming the resource model in Definition~\ref{def:resource-model}.
\end{theorem}

\providecommand{\bysame}{\leavevmode\hbox to3em{\hrulefill}\thinspace}
\providecommand{\MR}{\relax\ifhmode\unskip\space\fi MR }
\providecommand{\MRhref}[2]{%
  \href{http://www.ams.org/mathscinet-getitem?mr=#1}{#2}
}
\providecommand{\href}[2]{#2}

\newpage

\appendix
\section{Turing Completeness of \find + \mkdir without Back-references}
\label{sec:no_backref_tc}

We expand on the sketch of the proof in Section~\ref{sec:no_backref_tc_sketch}.

We prove Theorem~\ref{thm:no_backref_tc} by simulating 2-tag systems (Definition~\ref{def:2-tag-system}) with \find + \mkdir without regex back-references. We use the notation in the definition in the following proof. Note that the construction resembles the one in Section~\ref{sec:find_mkdir_tc}.

Let $\mu = 576$, and $\Sigma=\{s_1,\dots, s_\mu, s_{\mu+1} = H\}$. Let $A$ be a set of lowercase ASCII letters, and $\{s_1', \dots,s_\mu', s_{\mu+1}' \}\subset A^2$ be distinct strings with length 2. We will use \code{-regextype awk} for \find{}s with \code{-regex} to ensure that regexes are matched without back-references, while similar constructions should work for other regex types.

We define $\sigma_k = \code{/}\concat s_k'$, where \code{/} is the ASCII letter for directory separator. We let $\Sigma'=\{\sigma_1,\dots,\sigma_\mu, \sigma_{\mu+1} = \eta\}$ and $\phi: \Sigma \to \Sigma'$ be a function that maps $s_i$ to $\sigma_i$. We also define $\Phi: \Sigma^* \to \Sigma'^*$ so that $\Phi(\varepsilon) = \varepsilon$ and $\Phi(s_1\cdots s_k) = \Phi(s_1\cdots s_{k-1}) \concat \phi(s_k)$ for $k \ge 1$. We define $\pi(s_i) = \Phi(P(s_i))$ for $i \le \mu$.

Let $\sep = \code{\$\_)?(}$ . We first run the following command to embed the initial word.
\begin{lstlisting}
mkdir -p $\sep\concat\Phi(w_1)\concat$/$\concat\sep$
\end{lstlisting}

The idea for the following computation is to iteratively append the file path representing the next word to it using \sep as a separator. During the computation, the current working directory (CWD) should contain exactly one directory \sep, and under \sep there is exactly one empty directory. We call the path to the empty directory at a point of the computation the state of the system at that point.
After the computation halts, the state of the system should be $\sep\concat\Phi(w_1)\concat\code{/}\concat\sep\concat\Phi(w_2)\concat\code{/}\concat\sep\concat\cdots\concat\code{/}\concat\sep\concat\Phi(w_t)\concat\code{/}\concat\sep$, and during the computation the state should be a prefix of it. Let $\Theta_{i,j}$ be the following state. Our initial state is $\Theta_{1,0}$ and the last state will be $\Theta_{t,0}$.
$$\Theta_{i,j} = \cdots \concat\sep\concat\Phi(w_i)\concat\code{/}\concat\sep\concat\phi(x_{i+1,1})\cdots\phi(x_{i+1,j})$$
At $\Theta_{i, j}$, if $j=0$ and $w_i$ is a word to halt, we stop the computation. Otherwise, (1) if $\phi(x_{i,j+3})$ exists, we append $\phi(x_{i,j+3})$ to the state, and (2) otherwise (if we have copied the last alphabet of the previous word), we append $\pi(x_{i,1})\concat\code{/}\concat\sep$ to the state.

Let $\lambda = \code{/[a-z][a-z]}$ be a regex that matches any string in $\Sigma'$. Let $\Lambda=\code{(}\concat \lambda \concat \code{)*}$ be a regex that matches any repetition of strings in $\Sigma'$.
We also define $\Sep = \code{\bs\$\_\bs)\bs?\bs(}$ as a string that satisfies $\calL(\Sep) = \{\sep\}$. Define $t_1 = \code{1}, t_2 = \code{2}, \dots$ until $t_{\mu+1}$, which is an ASCII string representing $\mu+1$ in decimal.
Define the following strings.
\begin{align*}
  \alpha_k &= \code{.*} \concat\Sep \concat \lambda \concat \lambda \concat \code{(\braces)} \concat \sigma_k \concat \Lambda \concat \code{/}\concat\Sep\concat\Lambda\concat\code{/}\concat t_k \\
  \beta_k  &= \code{.*}\concat\Sep\concat \sigma_k \concat \Lambda \concat \code{/}\concat\Sep \concat \Lambda \\
  \gamma &= \code{.*}\concat\Sep\concat\code{(|} \concat \lambda \concat \code{|} \concat \eta \concat \Lambda \concat \code{)/}\concat\Sep \\
\end{align*}
We run the following as the second command. Newlines are for readability and have the same meaning as spaces.

\noindent
\begin{minipage}{\linewidth}
\begin{lstlisting}
find $\sep$ -regextype awk -type d -empty (
  -regex $\gamma$ -quit -o (*@\label{line:tc2:halt} \hfill $\triangleright$ \textrm{Halt check}@*)
  -execdir find -fprint {}/$\concat t_1$ -quit ; (*@\label{line:tc2:save} \hfill $\triangleright$ \textrm{Save $t_k$}@*)
  $\cdots$
  -execdir find -fprint {}/$\concat t_{\mu+1}$ -quit ;
  -exec find $\sep$ -regextype awk -type f -regex $\alpha_1$ -delete -quit ; (*@\label{line:tc2:check} \hfill $\triangleright$ \textrm{Check $\alpha_k$}@*)
  $\cdots$
  -exec find $\sep$ -regextype awk -type f -regex $\alpha_{\mu+1}$ -delete -quit ;
  ( (*@\label{line:tc2:append_start}@*)
    ! -execdir find {}/$\concat t_1$ -quit ; -execdir mkdir {}/$\sigma_1$ ; -o (*@\hfill $\triangleright$ \textrm{Append next}@*)
    $\cdots$
    ! -execdir find {}/$\concat t_{\mu+1}$ -quit ; -execdir mkdir {}/$\sigma_{\mu+1}$ ; -o
    -false
  ) -o ( (*@\label{line:tc2:append_end}@*)
    -regex $\beta_1$ -execdir mkdir -p {}$\concat\pi(s_1)\concat$/$\concat\sep$ ; -o (*@\label{line:tc2:transition} \hfill $\triangleright$ \textrm{Transition}@*)
    $\cdots$
    -regex $\beta_{\mu}$ -execdir mkdir -p {}$\concat\pi(s_{\mu})\concat$/$\concat\sep$ ; -o
    -printf unreachable
  ) ,
  -exec find $\sep$ -type f -delete ;
)
\end{lstlisting}
\end{minipage}

It uses some constructions that are not mentioned in the preliminaries section: \code{-type d} and \code{-type f} are evaluated to true if and only if the current file is a directory and a regular file, respectively. \code{-fprint \textit{file}} prints the path to the current file to \textit{file} (while we only use it to create a file). \code{\textit{expr1} , \textit{expr2}} evaluates both \textit{expr1} and \textit{expr2} and its value becomes the value of \textit{expr2}. The \code{,} has the least precedence (less than \code{-o}).

Finally we run the following command to output the encoded result.

\begin{lstlisting}
find $\sep$ -depth ! -empty -name $\sep$ -execdir find $\sep$ ! -name $\sep$ -printf /%f ; -quit
\end{lstlisting}
We show that the output of the last command is $\Phi(w_t)$ if the original 2-tag system halts with $w_t$, and the second command does not halt if the original does not halt, proving Theorem~\ref{thm:no_backref_tc}.

\begin{proof}
  Remember that $\Sigma' \subset \calL(\lambda), \Sigma'^* \subset \calL(\Lambda)$ and any string including a character in \sep or ending with \code{/} is not in $\calL(\lambda)$ nor $\calL(\Lambda)$. $\{\sigma_k\} = \calL(\sigma_k)$ and $\{\eta\} = \calL(\eta)$. We have $\Theta_{i,j} = \cdots\concat\sep\concat\phi(x_{i,1})\concat\cdots\concat\phi(x_{i,|w_i|})\concat\code{/}\concat\sep\concat\phi(x_{i+1,1})\concat\cdots\concat\phi(x_{i+1,j})$, where $x_{i+1,j+2} = x_{i,j}$ for $1 \le j\le |w_i|-2$.

  We prove that the state is $\Theta_{t,0}$ if the second command halts by induction. Let us call the expression in the outermost parentheses the \emph{main expression}.
  The \code{-type d} predicate is evaluated to true if and only if the current file is a directory. This means the expression is evaluated only when the current file is an empty directory.
  As the base case, \find visits $\Theta_{1,0}$ and evaluates the main expression for the first time. Assume the current state is $\Theta_{i,j}$ and \find visits it, evaluating the main expression. We also assume that before the evaluation $\Theta_{i,j}$ is the only empty directory under CWD, and there is no regular file (that matches \code{-type f}) under CWD, calling it the \emph{cleanness condition}. We will prove that after the evaluation the cleanness condition still holds as well.
  By construction, $\gamma$ matches and halts computation if and only if $j=0$ and $|w_i| < 2$ or $x_{i,1} = H$, meaning $w_i$ is the halting word (Line~\ref{line:tc2:halt}). If it does not halt, it creates marker files $\Theta_{i,j}\concat\code{/}\concat t_k$ for $k = 1,\dots,\mu+1$ (Line~\ref{line:tc2:save}). Then it evaluates the expressions from line~\ref{line:tc2:check}. The $k$-th expression from line~\ref{line:tc2:check} executes \find for the \sep in the CWD. Due to the cleanness condition and \code{-type f}, if it matches, it should be one of $t_1, \dots, t_{\mu+1}$. Also because $\alpha_k$ ends with $\code{/}\concat t_k$, it can match with $t_k$ only. If the regex matches $t_k$, it means
  $$\Theta_{i,j}\concat\code{/}\concat t_k\in \calL(\code{.*} \concat\Sep \concat \lambda \concat \lambda \concat \code{(}\concat\Theta_{i,j}\concat\code{)} \concat \sigma_k \concat \Lambda \concat \code{/}\concat\Sep\concat\Lambda\concat\code{/}\concat t_k)$$
  meaning
  $$\Theta_{i,j} \in \calL(\code{.*} \concat\Sep \concat \lambda \concat \lambda \concat \code{(}\concat\Theta_{i,j}\concat\code{)} \concat \sigma_k \concat \Lambda \concat \code{/}\concat\Sep\concat\Lambda)$$
  As explained in Section~\ref{sec:no_backref_tc_sketch},
  $$\calL(\code{(}\concat\Theta_{i,j}\concat\code{)}) = \phi(x_{i+1, 1})\cdots\phi(x_{i+1,j}) = \phi(x_{i,3})\cdots\phi(x_{i,j+2})$$

  Therefore, the match happens if and only if $\sigma_k = \phi(x_{i, j+3})$. That is, $t_k$ is deleted if and only if there exists a $\sigma_k = \phi(x_{i, j+3})$, and no marker files are deleted if there is no such $\sigma_k$, meaning $j\ge |w_i| - 2$. Whether deletion happens or not, the code block from line~\ref{line:tc2:append_start} is evaluated. This is evaluated to true creating a directory $\Theta_{i,j+1}$ if and only if there exists a deleted marker file $t_k$, because the $k$-th expression's \find exits with non-zero status code only if the starting point $\Theta_{i,j}\concat\code{/}\concat t_k$ does not exist. The code block from line~\ref{line:tc2:append_end} is therefore evaluated when $j = |w_i| - 2$ for the first time for each $i$. By construction this turns the state into $\Theta_{i+1, 0}$ (see the proof in Section~\ref{sec:find_mkdir_tc} for details). Finally, because of the \code{,} operator, the line \code{-exec find} \sep \code{-type f -delete ;} is always evaluated, deleting all marker files, restoring the cleanness condition.

  The third command (if the second command halts) will output $\Phi(w_t)$ as explained in the proof in Section~\ref{sec:find_mkdir_tc}.
\end{proof}

\section{Toy Examples}
\label{sec:toy-examples}

\subsection{\find + \mkdir with back-references}

We use the following 2-tag system instance as an example\footnote{Taken from \protect\url{https://en.wikipedia.org/wiki/Tag_system\#Example:_A_simple_2-tag_illustration}}. The expected output is \code{Hcccccca}.

\begin{align*}
  \Sigma &= \{\code{a}, \code{b}, \code{c}, H = \code{H}\} \\
  P(\code{a}) &= \code{ccbaH}\\
  P(\code{b}) &= \code{cca}\\
  P(\code{c}) &= \code{cc} \\
  w_1 &= \code{baa} \\
\end{align*}

Following shell script directly translates the construction in Section~\ref{sec:find_mkdir_tc}. If it is run, it outputs the encoded halting word \code{/ad/ac/ac/ac/ac/ac/ac/aa}.

\begin{lstlisting}[language=bash,mathescape=false,
  keywordstyle=\color{Blue},
  commentstyle=\color{ForestGreen},
  stringstyle=\color{BrickRed}]
#!/bin/bash
TMP="$(mktemp -d)"
cd "$TMP" # For safety.
# Start of the main program

A=("aa" "ab" "ac" "ad")
sigma=(/"${A[0]}" /"${A[1]}" /"${A[2]}" /"${A[3]}")
eta="${sigma[3]}"
lambda="/[a-z][a-z]"
Lambda='\('"$lambda"'\)*'

pi=(
    "${sigma[2]}${sigma[2]}${sigma[1]}${sigma[0]}${sigma[3]}"
    "${sigma[2]}${sigma[2]}${sigma[0]}"
    "${sigma[2]}${sigma[2]}"
)
Phiw1="${sigma[1]}${sigma[0]}${sigma[0]}" # baa

alpha=(
  '.*_'"$lambda$lambda"'\('"$Lambda"'\)'"${sigma[0]}$Lambda"'/_\1'
  '.*_'"$lambda$lambda"'\('"$Lambda"'\)'"${sigma[1]}$Lambda"'/_\1'
  '.*_'"$lambda$lambda"'\('"$Lambda"'\)'"${sigma[2]}$Lambda"'/_\1'
  '.*_'"$lambda$lambda"'\('"$Lambda"'\)'"${sigma[3]}$Lambda"'/_\1'
)
beta=(
  '.*_'"${sigma[0]}$Lambda"'/_'"$Lambda"
  '.*_'"${sigma[1]}$Lambda"'/_'"$Lambda"
  '.*_'"${sigma[2]}$Lambda"'/_'"$Lambda"
)
gamma='.*_\(\|'"$lambda"'\|'"$eta$Lambda"'\)/_'

mkdir -p '_'"$Phiw1"'/_'
find _ -empty \( \
    -regex "$gamma" -quit -o \
    -regex "${alpha[0]}" -execdir mkdir {}"${sigma[0]}" \; -o \
    -regex "${alpha[1]}" -execdir mkdir {}"${sigma[1]}" \; -o \
    -regex "${alpha[2]}" -execdir mkdir {}"${sigma[2]}" \; -o \
    -regex "${alpha[3]}" -execdir mkdir {}"${sigma[3]}" \; -o \
    -regex "${beta[0]}" -execdir mkdir -p {}"${pi[0]}"/_ \; -o \
    -regex "${beta[1]}" -execdir mkdir -p {}"${pi[1]}"/_ \; -o \
    -regex "${beta[2]}" -execdir mkdir -p {}"${pi[2]}"/_ \; -o \
    -printf unreachable \
\)
find _ -depth ! -empty -name _ -execdir find _ ! -name _ -printf /%f \; -quit

# End of the main program
rm -rf "$TMP"
\end{lstlisting}

\subsection{GitHub repository}

Other toy examples are available at the author's GitHub repository
\url{https://github.com/ogiekako/pub_find_mkdir}.


\begin{thebibliography}{10}

\bibitem{bratus2013weird}
Julian Bangert, Sergey Bratus, Rebecca Shapiro, and Sean~W. Smith, \emph{The page-fault weird machine: Lessons in instruction-less computation}, 7th USENIX Workshop on Offensive Technologies (WOOT 13), USENIX Association, 2013.

\bibitem{blaess-turing-sed}
Christophe Blaess, \emph{turing.sed}, \url{https://sed.sourceforge.net/grabbag/scripts/turing.sed}, 2001, Link dead. Archived from the original on 2018-01-16: \url{https://web.archive.org/web/20180116201401/http://sed.sourceforge.net/grabbag/scripts/turing.sed}. Accessed: 2026-01-10.

\bibitem{sed-tc}
Christophe Blaess{\relax}, \emph{Implementation of a {Turing} {Machine} as {Sed} script}, \url{https://sed.sourceforge.net/grabbag/scripts/turing.txt}, 2003, Link dead. Archived from the original on 2018-02-20: \url{https://web.archive.org/web/20180220011912/http://sed.sourceforge.net/grabbag/scripts/turing.txt}. Accessed: 2026-01-10.

\bibitem{churchill2019magic}
Alex Churchill, Stella Biderman, and Austin Herrick, \emph{{Magic}: The {Gathering} is {Turing} {Complete}}, 10th International Conference on Fun with Algorithms ({FUN} 2021), {LIPIcs}, vol. 157, Schloss Dagstuhl -- Leibniz-Zentrum f{\"u}r Informatik, 2021, pp.~11:1--11:21.

\bibitem{demol2008turing}
Liesbeth De~Mol, \emph{Tag systems and {Collatz}-like functions}, Theoretical Computer Science \textbf{390} (2008), no.~1, 92--101.

\bibitem{dolan2013mov}
Stephen Dolan, \emph{mov is {Turing}-complete}, \url{https://stedolan.net/research/mov.pdf}, 2013, Link dead. Archived from the original on 2021-02-14: \url{https://web.archive.org/web/20210214020524/https://stedolan.net/research/mov.pdf}. Accessed: 2026-01-03.

\bibitem{gnu-findutils}
{Free Software Foundation}, \emph{Gnu findutils manual}, 2024, Version 4.10.0.

\bibitem{hackaday2024find}
Hackaday, \emph{Proof that find + mkdir are {Turing}-{Complete}}, \url{https://hackaday.com/2024/08/05/proof-that-find-mkdir-are-turing-complete/}, 2024, Accessed: 2026-01-10.

\bibitem{minsky1961tag}
Marvin~L. Minsky, \emph{Recursive unsolvability of {Post}'s problem of ``{Tag}'' and other topics in theory of {Turing} machines}, Annals of Mathematics \textbf{74} (1961), no.~3, 437--455.

\bibitem{minsky1967computation}
Marvin~L.{\relax} Minsky, \emph{Computation: Finite and infinite machines}, Prentice-Hall, Inc., 1967.

\bibitem{neary2009four}
Turlough Neary and Damien Woods, \emph{Four small universal {Turing} machines}, Fundamenta Informaticae \textbf{91} (2009), no.~1, 123--144.

\bibitem{ogiekako2024find}
Keigo Oka, \emph{find + mkdir is {Turing} complete}, \url{https://ogiekako.vercel.app/blog/find_mkdir_tc}, 2024, Accessed: 2026-01-10.

\bibitem{post1943formal}
Emil~L. Post, \emph{Formal reductions of the general combinatorial decision problem}, American Journal of Mathematics \textbf{65} (1943), no.~2, 197--215.

\bibitem{rogozhin1996small}
Yurii Rogozhin, \emph{Small universal {Turing} machines}, Theoretical Computer Science \textbf{168} (1996), no.~2, 215--240.

\bibitem{posix-find}
{The Open Group}, \emph{The open group base specifications issue 8, 2024 edition: find}, 2024, Accessed: 2026-01-10.

\end{thebibliography}
\end{document}